\RequirePackage{amsthm}
\documentclass[sn-mathphys-num,pdflatex]{sn-jnl}
\usepackage{fix-cm}
 
\usepackage{graphicx}%
\usepackage{multirow}%
\usepackage{amsmath,amssymb,amsfonts}%
\usepackage{amsthm}%
\usepackage{mathrsfs}%
\usepackage[title]{appendix}%
\usepackage{xcolor}%
\usepackage{textcomp}%
\usepackage{manyfoot}%
\usepackage{booktabs}%
\usepackage{listings}%

\usepackage{tabularray}
\usepackage{overpic}

\theoremstyle{thmstyleone}%
\newtheorem{theorem}{Theorem}
\newtheorem{corollary}[theorem]{Corollary}%
\newtheorem{lemma}[theorem]{Lemma}%
\newtheorem{claim}[theorem]{Claim}%

\newtheorem{conjecture}[theorem]{Conjecture}%

\theoremstyle{thmstyletwo}%
\newtheorem{example}{Example}%
\newtheorem{remark}{Remark}%

\theoremstyle{thmstylethree}%

\usepackage{mathtools}

\renewcommand{\leq}{\leqslant}

\renewcommand{\geq}{\geqslant}

\newcommand{\cR}{\mathcal{R}}

\newcommand{\F}{\mathbb{F}}

\newcommand{\eqdef}{\triangleq}

\newcommand{\rc}{\mathrm{rc}}
\newcommand{\pal}{\mathrm{pal}}

\newcommand{\bs}{\textmd{-}}

\newif\ifshowproofs
\showproofstrue 

\begin{document}

\title[Optimal Functional $2^{s-1}$-Batch Codes: Exploring New Sufficient Conditions]{Optimal Functional $2^{s-1}$-Batch Codes: Exploring New Sufficient Conditions}


\author[1]{\fnm{Lev} \sur{Yohananov}}\email{lev.yo@migal.org.il}
\author[2]{\fnm{Isaac Barouch} \sur{Essayag}}\email{isaac.es@migal.org.il}

\affil[1]{\orgdiv{The Bioinformatics and Genomics Center}, \orgname{MIGAL -- Galilee Research Institute}, \orgaddress{\city{Kiryat Shmona}, \postcode{11016}, \country{Israel}}}
\affil[2]{\orgdiv{The Bioinformatics and Genomics Center}, \orgname{MIGAL -- Galilee Research Institute}, \orgaddress{\city{Kiryat Shmona}, \postcode{11016}, \country{Israel}}}


\abstract{
A functional $k$-batch code of dimension $s$ consists of $n$ servers storing linear combinations of $s$ linearly independent information bits. 
These codes are designed to recover any multiset of $k$ requests, each being a linear combination of the information bits, by $k$ disjoint subsets of servers. 
A recent conjecture suggests that for any set of $k = 2^{s-1}$ requests, the optimal solution requires $2^s-1$ servers.
This paper shows that the problem of functional $k$-batch codes is equivalent to several other problems. 
Using these equivalences, we derive sufficient conditions that improve understanding of the problem and enhance the ability to find the optimal solution.
}




\maketitle

\section{Introduction}

A \emph{batch} code encodes a string $x$ of length $s$ into $n$ strings stored on servers. 
This encoding ensures that any batch request for $m$ distinct bits (or more generally, symbols) from $x$ can be decoded by reading at most $t$ bits from each server. 
Batch codes were first introduced by Ishai et al.~\cite{Ishai2004}, motivated by various applications in load balancing for storage and cryptographic protocols. 
Later, batch codes were extended to the \emph{multiset} batch code version, where the number of requests may be greater than 1. 

A \textit{functional $k$-batch code} is a code that ensures that a multiset of $k$ linear combinations of the $s$ information bits (or generally, symbols) can be decoded by reading a bit from each server. 
Given $s$ and $k$, the goal is to find the smallest $n$ for which a functional $k$-batch code exists.
This problem was considered in several papers;~\cite{YamawakiKamabeLu2017,YohananovYaakobi2022,ZhangEtzionYaakobi2020}.

Formally, an $FB\bs(n,s,k)$ functional $k$-batch  code (and in short $FB\bs(n,s,k)$ code) of dimension $s$ consists of $n$ servers storing linear combinations of $s$ linearly independent information bits. 
Any multiset of size $k$ of linear combinations from the linearly independent information bits, can be recovered by $k$ disjoint subsets of servers.
For example, assume that the binary string $x$ is of length $s=2$, denoted by $x = (x_1, x_2)$. For $n=3$, let  $s_1, s_2, s_3$ be binary servers store the following linear combinations:
\begin{align*}
      s_1 &= x_1, \\
      s_2 &= x_2, \\
      s_3 &= x_1 + x_2.
\end{align*}
Any $k=2$ linear combination requests must be decoded. 
Indeed, the request $(x_1, x_1)$ is decoded by $(s_1, s_2 + s_3)$, since $x_1=s_1,x_2=s_2+s_3$, and the request $(x_1, x_2)$ is decoded by $s_1, s_2$, since $x_1=s_1,x_2=s_2$. 
The other cases are left to the reader to verify readily.

The functional $k$-batch problem can be represented in an alternative, more convenient way.
The $k$ requests will be vectors $v_0, v_1, \dots, v_{k-1}$ (not necessarily distinct) of $s$ bits (or generally symbols).
The $n$ servers $s_1, s_2, \dots, s_n$ are also vectors of $s$ bits.
The goal is to find a partition of servers to sets $S_1,S_2,\dots,S_k$ such that the sum (or generally linear combination) of all servers in set $i$ equals the \emph{request vector} $v_i$.  
Every set $S_i$ is called a \emph{recovery set}.
A partition $S_1, S_2,\dots, S_k$ is called a \emph{solution} for the $k$ request vectors. 
In the previous example, the servers are:
\begin{align*}
    s_1 &= (1, 0), \\
    s_2 &= (0, 1), \\
    s_3 &= (1, 1),
\end{align*}
and the requests might be $v_1 = (1, 0)$, $v_2 = (1, 0)$.

The value $FB(s, k)$ is defined to be the minimum number of servers required for the existence of an $FB\bs(n,s,k)$ code.
Codes achieving $n=FB(s, k)$ are called \emph{optimal}.
We will see later that this $FB\bs(n=3,s=2,k=2)$ code presented above is optimal.
An open conjecture says that given $s$ and $k=2^{s-1}$, it holds that $FB(s, 2^{s-1}) = 2^s-1$.
\begin{conjecture}\label{conj_1}
    For any $s$, the value of $FB(s, 2^{s-1})$ is equal to $2^s - 1$.
\end{conjecture}

In~\cite{YohananovYaakobi2022}, it was shown that under the conditions $v_i \neq 0$ and $\sum_{i=1}^k v_i = 0$, a solution can be efficiently computed if all $v_i$'s are equal or if each $v_i$ has an odd Hamming weight. 
Additionally, an optimal algorithm was provided for the case where each element of $\mathbb{F}_2^s \setminus \{0\}$ appears exactly twice as a server, i.e., $n = 2^{s+1} - 2$ and $k = 2^s$.
In~\cite{Kovacs2023}, it was proven that a solution exists if the number of distinct values among the difference vectors is at most $s - 2\log s - 1$.


In this paper, we use an approach similar to that in~\cite{YohananovYaakobi2022}, namely, making all possible non-zero vectors of length $s$ serve as servers (which are exactly $2^s-1$ servers). 
For convenience, we add a virtual zero-server (a server equal to the zero vector of length $s$), so overall $\{s_1, s_2, \dots, s_n\} = \mathbb{F}_2^s$.
Generally, a response to any request can use any number of servers. However, our approach demonstrates that every request will use exactly two servers. 
In other words, assuming $\sum_{i=1}^k v_i = 0$, our goal is to show that every recovery set satisfies $|S_i| = 2$. 
This setup aligns with the conjecture proposed by Balister, Gy\H{o}ri, and Schelp~\cite{BalisterGyoriSchelp2011}.

\begin{conjecture}\label{conj_2}
    Let $s \geq 2, q=2^s$ be an integer and $k = 2^{s-1}$.
    If the nonzero difference vectors $v_1, v_2, \dots, v_k$ are given in $\mathbb{F}_q$ such that $\sum_{i=1}^k v_i = 0$ (and the $v_i$'s are not necessarily distinct), then $\mathbb{F}_2^s$ can be subdivided into disjoint pairs $\{x_i, y_i\}$ such that $x_i + y_i = v_i$ holds for every $i$.
\end{conjecture}
In this paper, we will show that proving Conjecture~\ref{conj_2} also proves Conjecture~\ref{conj_1}.


An equivalent problem can be formulated as follows: Given non-zero vectors 
$v = (v_1, v_2, \dots, v_k) \in \mathbb{F}_q^k$ for $k = 2^{s-1}, q = 2^s$, such that $\sum_{i=1}^k v_i = 0$, 
determine whether there exist $\alpha_1, \alpha_2, \dots, \alpha_k \in \mathbb{F}_q$ such that the matrix
\[
H = \begin{bmatrix}
    1 & \dots & 1 & 1 & \dots & 1 \\
    \alpha_1 & \dots & \alpha_k & \alpha_1 + v_1 & \dots & \alpha_k + v_k \\
    \alpha_1^2 & \dots & \alpha_k^2 & (\alpha_1 + v_1)^2 & \dots & (\alpha_k + v_k)^2 \\
    \vdots & \ddots & \vdots & \vdots & \ddots & \vdots \\
    \alpha_1^{k-1} & \dots & \alpha_k^{k-1} & (\alpha_1 + v_1)^{k-1} & \dots & (\alpha_k + v_k)^{k-1}
\end{bmatrix}
\]
has full rank.
If we treat $\alpha_i$'s as variables $x_i$, this matrix has a characteristic polynomial
\[
f_v(x_1, \dots, x_k) = v_1v_2\cdots v_k \prod_{1 \leq i < j \leq k} (x_i + x_j)(x_i + x_j + v_i)(x_i + x_j + v_j)(x_i + x_j + v_i + v_j).
\]
As shown in~\cite{MacWilliamsSloane}, $H$ is non-singular if and only if we can find 
$\alpha_1, \alpha_2, \dots, \alpha_k \in \mathbb{F}_q$ such that $f_v(\alpha_1, \alpha_2, \dots, \alpha_k) \neq 0$. 
We will show that Conjecture~\ref{conj_2} is true if and only if we can find 
$\alpha_1, \alpha_2, \dots, \alpha_k \in \mathbb{F}_q$ such that $f_v(\alpha_1, \alpha_2, \dots, \alpha_k) \neq 0$. 
This polynomial can be expressed by
\begin{align*}
    f_v(x_1, \dots, x_k) = v_1v_2\cdots v_k \prod_{1 \leq i < j \leq k} (x_i + x_j)^4 + \text{Terms of smaller degree},
\end{align*}
where in the $\prod_{1 \leq i < j \leq k} (x_i + x_j)^4$ we might focus on the monomial $\prod_{i=1}^k x_i^{q-2}$.
By Dyson's conjecture~\cite{Dyson1962}, it is known that the coefficient of this monomial is $\binom{q}{2, 2, \dots, 2}$. 
However, since the characteristic of $\mathbb{F}_q$ is $2$, this monomial vanishes, making the problem significantly more challenging.

In this paper, we study the polynomial $f_v$, providing several sufficient conditions for solving Conjecture~\ref{conj_2}. 
In~\cite{YamawakiKamabeLu2017}, a proof was presented for the cases $s=2,3,4$. 
However, the proof is somewhat cumbersome, as it requires the reader to go through many special cases. 
We will show that our techniques solve this conjecture for $s=2$ and $s=3$ in a more elegant manner. 
Furthermore, as will be explained later, our technique can also provide a solution for larger values of $s$ using a computer search.

The paper is organized as follows. 
In Section~\ref{sec:pre} we give all the necessary notation and definitions. 
In Section~\ref{sec:gen_case} we analyze several properties of $f_v$ and give sufficient conditions for solving Conjecture~\ref{conj_2}. 
Section~\ref{sec:2_3} provides a solution for $s=2,3$ cases. 
We conclude in Section~\ref{sec:conc} with a summary of the results and some open questions.

\section{Preliminaries}
\label{sec:pre}

In this section, we will show how to represent the problem described in Conjecture~\ref{conj_2} in another way that uses a finite ring of multivariable polynomials.
Throughout this paper, let $q = 2^s$, and $v_1, v_2, \dots, v_k$ be non-zero binary not necessarily distinct vectors of length $s$.
Let $k=2^{s-1}$.
A solution for $v=(v_1, v_2, \dots, v_k)$ is a partition of servers $\mathbb{F}_q$ into $S_1, S_2, \dots, S_k$ where each recovery set satisfies $|S_i| = 2$.
Let $f_v \in \mathbb{F}_q[x_1, x_2, \dots, x_k]$ be the following polynomial,
\[
f_v(x_1, \dots, x_k) = v_1v_2\cdots v_k\prod_{1 \leq i < j \leq k} (x_i + x_j)(x_i + x_j + v_i)(x_i + x_j + v_j)(x_i + x_j + v_i + v_j).
\]
 Note that \( f_v \) depends on \( v = (v_1, v_2, \dots, v_k) \).
    
Let $I = \langle x_1^q + x_1, x_2^q + x_2, \dots, x_k^q + x_k \rangle$ be an ideal of polynomials over $\mathbb{F}_q[x_1, x_2, \dots, x_k]$.
Let $R = \frac{\mathbb{F}_q[x_1, x_2, \dots, x_k]}{I}$ be a ring of polynomials and denote by ``$\equiv_R$" a congruence equivalence operator in $R$.

Noga Alon shows the following theorem in \cite{Alon1999}:
\begin{theorem}\label{theo:3}
     Let $\mathbb{F}$ be an arbitrary field, and let $f = f(x_1, \dots , x_n)$ be a polynomial in $\mathbb{F}[x_1, \dots, x_n]$.
     Let $S_1, \dots , S_n$ be nonempty subsets of $\mathbb{F}$ and define $g_i(x_i) = \prod_{s \in S_i} (x_i - s)$.
     If $f$ vanishes over all the common zeros of $g_1, \dots , g_n$ (that is, if $f(s_1, \dots , s_n) = 0$ for all $s_i \in S_i$), then there are polynomials
     $h_1, \dots, h_n \in \mathbb{F}[x_1, \dots, x_n]$ satisfying $\deg(h_i) \leq \deg(f) - \deg(g_i)$ so that
        \begin{align*}
            f = \sum_{i=1}^n h_i g_i.
        \end{align*}
        Moreover, if $f, g_1, \dots, g_n$ lie in $R[x_1, \dots , x_n]$ for some subring $R$ of $\mathbb{F}$, then there are polynomials  $h_i \in R[x_1, \dots , x_n]$ as above.
\end{theorem}

We now extend this result to a similar setting where $x_1$ is dependent on other variables rather than constants.
\begin{theorem}\label{theo:4Lev}
     Let $\mathbb{F}$ be an arbitrary field, and let $f = f(x_1, \dots , x_n)$ be a polynomial in $\mathbb{F}[x_1, \dots, x_n]$.
     Let $S_1,S_2, \dots , S_{n}$ be nonempty subsets of $\mathbb{F}$.
     Define $g_1(x_1,\dots,x_n) = \prod_{\alpha \in S_1}(x_1-(\sum^n_{i=2}x_i+\alpha))$, and for all 
     $2\leq i\leq n$, $g_i(x_i) = \prod_{s \in S_i} (x_i - s)$.
     If $f$ vanishes over all the common zeros of $g_1, \dots , g_n$ (that is, if $f(s_1, \dots , s_n) = 0$ for all $s_i \in S_i,2\leq i \leq n$ and $s_1-\sum^n_{i=2}s_i\in S_1$), then there are polynomials
     $h_1, \dots, h_n \in \mathbb{F}[x_1, \dots, x_n]$ satisfying $\deg(h_i) \leq \deg(f) - \deg(g_i)$ so that
        \begin{align*}
            f = \sum_{i=1}^n h_i g_i.
        \end{align*}
\end{theorem}
\begin{proof}
    Substitute $y_1 = x_1-\sum^n_{i=2}x_i$, and use Theorem~\ref{theo:3}.
\end{proof}

We illustrate this result with an example.
\begin{example}\label{examp:1}
    Let $f(v_1,v_2)=v_1v_2(v_1^2v_2+v_1v^2_2+1)$ be a polynomial in $GF(4)[v_1,v_2]$.
    It can be verified that this polynomial vanishes if $v_1\neq v_2$.
    Thus if we take arbitrary $v_2 \in GF(4)$ and after that we take $v_1\in GF(4)\setminus \{v_2\}$ then $f(v_1,v_2)=0$.
    According to this, we define $S_2 = GF(4) $ and  $S_1 = GF(4)\setminus\{0\}$, and two polynomials, $g_2(v_1,v_2)=v^4_2+v_2$,
    \begin{align*}
        g_1(v_1,v_2)=(v_1+v_2+1)(v_1+v_2+\alpha)(v_1+v_2+\alpha^2)= (v_1+v_2)^3+1.
    \end{align*}
    Due to Theorem~\ref{theo:4Lev}, $f$ can be represented as a linear combination of $g_1,g_2$.
     Indeed the polynomial $f$ can be represented as follows
    \begin{align*}
        f(v_1,v_2) = \Big((v_1+v_2)^3+1\Big)v^2_2+\Big(v^4_2+v_2\Big)(v_1+v_2).
    \end{align*}
    \begin{proof}
    \begin{align*}
        f(v_1,v_2) &= \Big((v_1+v_2)^3+1\Big)v^2_2+\Big(v^4_2+v_2\Big)(v_1+v_2)\\
                   &= (v_1^3+v_2^3+v_1^2v_2+v_1v_2^2+1)v_2^2+v_2^5+v_2^2+v_1v_2^4+v_1v_2\\
                   & = v_1^3v_2^2+v_2^5+v_1^2v_2^3+v_1v_2^4+v_2^2+v_2^5+v_2^2+v_1v_2^4+v_1v_2\\
                   & = v_1v_2(v_1^2v_2+v_1v^2_2+1).
    \end{align*}
    \end{proof}
\end{example}


\begin{remark}
    In the proof of Theorem~\ref{theo:3}, the constructed algorithm ensures that no variable $x_i$ has an exponent $t_i \geq |S_i|+1$ in any of the $h_i$'s. 
    This property also holds for Theorem~\ref{theo:4Lev} and will be used in Section~\ref{sec:gen_case}.
\end{remark}

Another useful theorem is shown and will be used in this paper.
\begin{theorem}\label{theo:3_1}
    Let $F$ be an arbitrary field, and let $f = f(x_1, \dots, x_n)$ be a polynomial in $F[x_1, \dots, x_n]$. 
Suppose the degree $\deg(f)$ of $f$ is $\sum_{i=1}^n t_i$, where each $t_i$ is a nonnegative integer, and suppose the coefficient of $\prod_{i=1}^n x_i^{t_i}$ in $f$ is nonzero. Then, if $S_1, \dots, S_n$ are subsets of $F$ with $|S_i| > t_i$, there exist $s_1 \in S_1, s_2 \in S_2, \dots, s_n \in S_n$ such that 
\[
f(s_1, \dots, s_n) \neq 0.
\]
\end{theorem}

\begin{theorem}\label{theo:5}
    Let \( v = (v_1, v_2, \dots, v_k) \). The following statements are equivalent:
    \begin{enumerate}
        \item There is a solution for \( v \). 
        \item There exist \( \alpha_1, \dots, \alpha_k \in \mathbb{F}_q \) such that \( f_v(\alpha_1, \dots, \alpha_k) \neq 0 \).
        \item It holds that \( f_v \not\equiv_R 0 \).
    \end{enumerate}
\end{theorem}

\begin{proof}
    \textbf{1 $\Rightarrow$ 2:} Suppose there is a solution for the requests \( v_1, \dots, v_k \). 
    This means that there are pairwise distinct elements \( \alpha_1, \dots, \alpha_k, \beta_1, \dots, \beta_k \in \mathbb{F}_q \) such that for all \( 1 \leq i \leq k \),
    \[
         \alpha_i + \beta_i = v_i.
    \]
    Note that since \( k = 2^{s-1} \), \( \{ \alpha_1, \dots, \alpha_k, \beta_1, \dots, \beta_k \} = \mathbb{F}_q \).
    We now examine the polynomial function \( f_v \) evaluated at \( \alpha_1, \dots, \alpha_k \). The function is given by:
    \[
    f_v(\alpha_1, \dots, \alpha_k) = \prod_{1 \leq i < j \leq k} (\alpha_i + \alpha_j)(\alpha_i + \alpha_j + v_i)(\alpha_i + \alpha_j + v_j)(\alpha_i + \alpha_j + v_i + v_j)
    \]
    \[
    = \prod_{1 \leq i < j \leq k} (\alpha_i + \alpha_j)(\beta_i + \alpha_j)(\alpha_i + \beta_j)(\beta_i + \beta_j)
    \]
    Since \( \alpha_1, \dots, \alpha_k, \beta_1, \dots, \beta_k \) are distinct, the product of terms in the above expression is non-zero. Thus, we have:
    \[
    f_v(\alpha_1, \dots, \alpha_k) \neq 0.
    \]

    \textbf{2 $\Rightarrow$ 3:} Suppose there exist \( \alpha_1, \dots, \alpha_k \in \mathbb{F}_q \) such that \( f_v(\alpha_1, \dots, \alpha_k) \neq 0 \), but assume for contradiction that \( f_v \equiv_R 0 \). By the definition of congruence in \( R \), this means that \( f_v \in I \), where \( I = \langle x_1^q + x_1, x_2^q + x_2, \dots, x_k^q + x_k \rangle \). Since \( f_v \in I \), we can express \( f_v \) as a linear combination of the generators of \( I \):
    \[
    f_v = \sum_{i=1}^k h_i (x_i^q + x_i),
    \]
    where \( h_1, h_2, \dots, h_k \in \mathbb{F}_q[x_1, \dots, x_k] \). Evaluating this at \( \alpha_1, \dots, \alpha_k \):
    \[
    f_v(\alpha_1, \dots, \alpha_k) = \sum_{i=1}^k h_i (\alpha_i^q + \alpha_i) = 0.
    \]
    Since each term \( (\alpha_i^q + \alpha_i) = 0 \) for \( \alpha_i \in \mathbb{F}_q \), the entire sum vanishes. This contradicts the assumption that \( f_v(\alpha_1, \dots, \alpha_k) \neq 0 \), hence \( f_v \not\equiv_R 0 \).

    \textbf{3 $\Rightarrow$ 1:} Suppose that \( f_v \not\equiv_R 0 \), but there is no solution for \( v_1, v_2, \dots, v_k \). If there is no solution, it means that for any \( \alpha_1, \dots, \alpha_k \) and \( \beta_1, \dots, \beta_k \) such that \( \alpha_i + \beta_i = v_i \), the set \( \{ \alpha_1, \dots, \alpha_k, \beta_1, \dots, \beta_k \} \) is not equal to \( \mathbb{F}_q \). This implies that there is repetition in the set, leading to the fact that \( f_v(\alpha_1, \dots, \alpha_k) = 0 \),
    for any choice of \( \alpha_1, \alpha_2, \dots, \alpha_k \).
    By Theorem~\ref{theo:3}, in this case, we have \( f_v \) represented as
    \[
    f_v = \sum_{i=1}^k h_i (x_i^q + x_i),
    \]
    i.e., \( f_v \equiv_R 0 \). 
    This contradicts the assumption that \( f_v \not\equiv_R 0 \), and hence a solution must exist for \( v \).

    Thus, we have shown that the statements are equivalent.
\end{proof}

Our next goal is to demonstrate that the necessary condition for a solution to exist for \( v \) is that \( \sum_{i=1}^k v_i = 0 \) (for non-zero $v_i$'s).
\begin{lemma}\label{lemma:6}
     If \( \sum_{i=1}^k v_i \neq 0 \), then \( f_v \equiv_R 0 \).
\end{lemma}

\ifshowproofs
\begin{proof}
    Assume to the contrary that \( f_v \not\equiv_R 0 \).
    By Theorem~\ref{theo:5}, since \( f_v \not\equiv_R 0 \), there exists a solution for \( v \). 
    This means there are pairwise distinct elements \( \alpha_1, \dots, \alpha_k, \beta_1, \dots, \beta_k \in \mathbb{F}_q \) such that for all \( 1 \leq i \leq k \),
    \[
    \alpha_i + \beta_i = v_i.
    \]
    Since \( k = 2^{s-1} \), and since \( \alpha_i \) and \( \beta_i \) are distinct, it follows that
    \[
    \{ \alpha_1, \dots, \alpha_k, \beta_1, \dots, \beta_k \} = \mathbb{F}_q,
    \]
    which implies
    \[
    \sum_{i=1}^k (\alpha_i + \beta_i) = \sum_{i=1}^k v_i.
    \]
    However, we are given that \( \sum_{i=1}^k v_i \neq 0 \), so this leads to a contradiction since 
    \[
    0 = \sum_{i=1}^k (\alpha_i + \beta_i) = \sum_{i=1}^k v_i.
    \]
    Thus, we conclude that \( f_v \equiv_R 0 \).
\end{proof}
\else
\textcolor{red}{Proof omitted.}
\fi

To prove Conjecture~\ref{conj_2}, we are left to show that the condition \( \sum_{i=1}^k v_i = 0 \) is not only necessary but also sufficient. 
By Theorem~\ref{theo:5}, it is sufficient to prove the following theorem.
\begin{theorem}~\label{theo:main2}
      If $\sum^k_{i=1}v_i = 0$ then $f_v \not\equiv_R 0$.
\end{theorem}

\ 

Denote by $w_H(v_i)$ the Hamming weight of the binary representation of $v_i$. 
In~\cite{YohananovYaakobi2022}, it was proven that if all $v_i$'s are equal, or if all $v_i$'s satisfy $w_H(v_i)$ is odd (under the assumption that the $v_i$'s are non-zero and $\sum_{i=1}^k v_i = 0$), then there is an efficient algorithm that finds a solution.
In~\cite{Kovacs2023}, it was proven that there is a solution if the number of distinct values among the difference vectors is at most $s - 2 \log s - 1$.
In other words, we can conclude the following corollary:

\begin{corollary}\label{cor:8}
    If all $v_i$'s are equal, or if all $v_i$'s satisfy $w_H(v_i)$ is odd, or the number of distinct values among the difference vectors is at most $s - 2 \log s - 1$, then $f \not\equiv_R 0$.
\end{corollary}

\section{The Analysis of \( f_v \)}
\label{sec:gen_case}

Let \( k = 2^{s-1} \) and \( q = 2^s \).
Let \( v = (v_1, v_2, \dots, v_k) \) be non-zero binary (not necessarily distinct) vectors of length \( s \).
The polynomial can be formulated as follows:
\[
f_v(x_1, x_2, \dots, x_k) = v_1 \cdots v_k \prod_{1 \leq i < j \leq k} (x_i + x_j)(x_i + x_j + v_i)(x_i + x_j + v_j)(x_i + x_j + v_i + v_j)
\]
\[
\equiv_R \sum_{i_1, i_2, \dots, i_k} g_j(v_1, v_2, \dots, v_k) x_1^{i_1} x_2^{i_2} \cdots x_k^{i_k} \eqdef \bar{f_v}(x_1, x_2, \dots, x_k),
\]
where \( g_j \in \mathbb{F}_q \) is the \( j \)-th coefficient of the polynomial \( \bar{f_v} \), which is congruent to \( f_v \) in \( R \).
Since \( \bar{f_v} \in R \), the exponents \( i_j \)'s are bounded by \( q-1 \).
Note that \( g_j \) depends on \( v_1, v_2, \dots, v_k \), and the monomials are arranged lexicographically with respect to the indices \( i_1, i_2, \dots, i_k \).
By abuse of notation, we sometimes refer to \( g_j(v_1, v_2, \dots, v_k) \) as a polynomial in \( \mathbb{F}_2[v_1, v_2, \dots, v_k] \).

By Dyson's conjecture~\cite{Dyson1962}, the coefficient of the monomial \( \prod_{i=1}^k x_i^{q-2} \) in the  $\prod_{1 \leq i < j \leq k} (x_i + x_j)^4$ part is \( \binom{q}{2, 2, \dots, 2} \), which is even.
Consequently, for this monomial, \( g_j = 0 \), making the problem more challenging.
However, by Corollary~\ref{cor:8}, not all \( g_j \)'s are zero.
For example, we will prove that the \( g_j \) of \( x_1^{q-2} \prod_{i=1}^{k-1} x_{i+1}^{2q-4i} \) is non-zero.
Furthermore, every non-zero \( g_j \) depends on all \( v_i \)'s and satisfies \( \deg_{v_i}(g_j) = q-1 \), as will be shown later.

\begin{claim}\label{claim:9}
    The \( g_j \) of the monomial \( x_1^{q-2} \prod_{i=1}^{k-1} x_{i+1}^{2q-4i} \) is non-zero.
\end{claim}

\begin{proof}
    Note that the degree of \( f_v \) is:
    \[
    \deg(f_v) = 4 \cdot \binom{k}{2} = 2k(k-1) = q\left(\frac{q}{2} - 1\right) = \frac{q}{2}(q-2).
    \]

    Consider the specific monomial of \( f_v \):
    \[
    g_j(v_1, v_2, \dots, v_k)x_1^{q-2} \prod_{i=1}^{k-1} x_{i+1}^{2q-4i}.
    \]

    The \( v_1^{q-1}x_1^{q-2} \) part is uniquely obtained from:
    \[
    v_1\prod_{j=2}^k (x_1+x_j)(x_1+x_j+v_1)(x_1+x_j+v_2)(x_1+x_j+v_1+v_j),
    \]
    since there are \( 2(k-1) + 1 = q-1 \) occurrences of \( v_1 \), and after selecting these, there are \( 2(k-1) = q-2 \) occurrences of \( x_1 \).
    Similarly, for all $1\leq i \leq k-1$, the \( x_{i+1}^{2q-4i} \) part is uniquely obtained from the remaining terms:
    \[
    v_{i+1}\prod_{j=i+2}^k (x_{i+1}+x_j)(x_{i+1}+x_j+v_{i+1})(x_{i+1}+x_j+v_j)(x_{i+1}+x_j+v_{i+1}+v_j),
    \]
    since there are \( 4(k-i) = 2q-4i \) occurrences of \( x_{i+1} \).
\end{proof}

Denote \( B = \{j \mid g_j \neq 0\} \).
Let \( A = \{(\beta_1, \beta_2, \dots, \beta_k) \mid \beta_i \in \mathbb{F}_q, \sum_{i=1}^k \beta_i \neq 0\} \).
Define \( P_k \) as the set of all polynomials \( g(v_1, v_2, \dots, v_k) \in \mathbb{F}_2[v_1, v_2, \dots, v_k] \) such that \( 1 \leq \deg_{v_i}(g) \leq q-1 \) for all \( i \), and \( g(\beta_1, \beta_2, \dots, \beta_k) = 0 \) for all \( (\beta_1, \beta_2, \dots, \beta_k) \in A \).

We start with the following claim.
\begin{claim}\label{claim:9_1}
    For all $g_j \neq 0$, there exists a point $\beta_1, \dots, \beta_k \in \F_q$ such that $g_j(\beta_1, \dots, \beta_k) \neq 0$.
    Moreover $\sum^k_{i=1}\beta_i = 0$.
\end{claim}
\begin{proof}
    It can be verified that for all \( i \), \( \deg_{v_i}(g_j) \leq q-1 \). 
    Let $\prod_{i=1}^k v_i^{t_i}$ be a term of maximal degree in $g_j$. 
    Let $S_i \subseteq \F_q$ be arbitrary sets for \( 1 \leq i \leq k \) such that \( |S_i| > t_i \). 
    By Theorem~\ref{theo:3_1}, there exists a point \(\beta_1 \in S_1, \dots, \beta_k \in S_k\) such that \( g_j(\beta_1, \dots, \beta_k) \neq 0 \).
    By Lemma~\ref{lemma:6}, it is deduced that \( \sum_{i=1}^k \beta_i = 0 \).
\end{proof}

Next, we proceed to the following important claim.
\begin{claim}\label{claim:10}
    For all \( j \in B \), \( g_j \in P_k \).
\end{claim}

\begin{proof}
    As mentioned before, for all $i,j$ it holds that $\deg_{v_i}(g_j)\leq q-1$.
    Assume, without loss of generality, that \( g_j \) is independent of \( v_k \).
    By Claim~\ref{claim:9_1}, there must exist \( \beta_1, \beta_2, \dots, \beta_k \in \mathbb{F}_q \) such that \( g_j(\beta_1, \beta_2, \dots, \beta_k) \neq 0 \) and \( \sum_{i=1}^k \beta_i = 0 \).
    However, since \( g_j \) is independent of \( v_k \), there exists a non-zero \( \gamma \neq \beta_k \) such that
    \[
    g_j(\beta_1, \beta_2, \dots, \beta_{k-1}, \beta_k) = g_j(\beta_1, \beta_2, \dots, \beta_{k-1}, \gamma) \neq 0.
    \]
    Therefore, for \( w = (\beta_1, \beta_2, \dots, \beta_{k-1}, \gamma) \), we obtain \( f_w \not\equiv_R 0 \), even though
    \[
    \sum_{i=1}^{k-1} \beta_i + \gamma = \beta_k + \gamma \neq 0.
    \]
    This contradicts Lemma~\ref{lemma:6}, which asserts that \( f_w \not\equiv_R 0 \) only if \( \sum_{i=1}^k \beta_i = 0 \).
    Hence, \( g_j \) must depend on all the variables \( v_i \), and therefore, \( g_j \in P_k \).
\end{proof}

In the next lemma, we want to show that every non-zero $g_j$ must satisfy \( \deg_{v_i}(g_j) = q-1 \).
\begin{lemma}\label{lemma:10}
    For all \( j \in B \) and for all \( 1 \leq i \leq k \), it holds that \( \deg_{v_i}(g_j) = q-1 \).
\end{lemma}

\begin{proof}
    From Claim~\ref{claim:10}, we know that there exists \( j \in B \) such that \( g_j \in P_k \).
    For the sake of contradiction, suppose that there exists \( i \) such that \( \deg_{v_i}(g_j) < q-1 \).
    Without loss of generality, assume that \( i = 1 \).
    By Claim~\ref{claim:9_1} there exist \( \beta_1, \beta_2, \dots, \beta_k \in \mathbb{F}_q \), with \( \sum_{i=1}^k \beta_i = 0 \), such that \( g_j(\beta_1, \beta_2, \dots, \beta_k) \neq 0 \).
    
    Consider the polynomial
    \[
    g(v_1) = g_j(v_1, \beta_2, \dots, \beta_k).
    \]
    We know that \( 1 \leq \deg_{v_1}(g_j) \leq q-2 \), meaning \( \deg(g) \leq q-2 \).
    Since \( g_j(\beta_1, \beta_2, \dots, \beta_k) \neq 0 \), we also have \( g(v_1) \neq 0 \), which implies \( 1 \leq \deg(g) \leq q-2 \).
    
    Since the degree of \( g \) in \( v_1 \) is strictly smaller than \( q-1 \), there must exist another point (in addition to \( \beta_1 \)) \( \gamma \in \mathbb{F}_q \) where \( g \) does not vanish.
    Given that \( \sum_{i=1}^k \beta_i = 0 \), it follows that \( \gamma + \sum_{i=2}^k \beta_i \neq 0 \).
    Thus, we obtain $ g(\gamma) = g_j(\gamma, \beta_2, \dots, \beta_k) \neq 0$.
    However, this contradicts Lemma~\ref{lemma:6}, which asserts that \( g_j \) may not vanish only if \( \gamma + \sum_{i=2}^k \beta_i = 0 \).

    Therefore, the degree of \( g_j \) in \( v_1 \) must be \( q-1 \), and similarly for all other variables. Hence, \( \deg_{v_i}(g_j) = q-1 \) for all \( i \).
\end{proof}

We proceed with the following important lemma.

\begin{lemma}\label{lemma:12}
    Every \( g_j(v_1, \dots, v_k) \) can be expressed as follows:
    \[
    g_j(v_1, \dots, v_k) = \Big((v_1 + \dots + v_k)^{q-1} + 1\Big)\Big(v_2 + \dots + v_k\Big)h_1 + \sum_{j=2}^k \Big(v_j^q + v_j\Big)h_j,
    \]
    where \( h_1, \dots, h_k \in \mathbb{F}_q[v_1, \dots, v_k] \) satisfy \( \deg(h_i) \leq \deg(g_j) - q \). Moreover, \( h_1 \) is independent of \( v_1 \).
\end{lemma}

\begin{proof}
    Let \( S_1 = \text{GF}(q) \setminus \{0\} \) and \( S_i = \text{GF}(q) \) for all \( 2 \leq i \leq k \). Define \( p_i(v_i) = v_i^q + v_i \) and
    \[
    p_1(v_1, \dots, v_k) = (v_1 + \dots + v_k)^{q-1} + 1.
    \]
    By Lemma~\ref{lemma:6}, if \( v_1 + \dots + v_k \neq 0 \), then \( f_v \equiv_R 0 \).
    Therefore, by Theorem~\ref{theo:4Lev}, every \( g_j \) can be expressed as:
    \[
    g_j(v_1, \dots, v_k) = \Big((v_1 + \dots + v_k)^{q-1} + 1\Big)h'_1 + \sum_{j=2}^k \Big(v_j^q + v_j\Big)h'_j,
    \]
    where \( h'_1, \dots, h'_k \in \mathbb{F}_q[v_1, \dots, v_k] \) satisfy \( \deg(h'_i) \leq \deg(g_j) - (q-1) \).

    If \( h'_1 \) depends on \( v_1 \), then
    \[
    \deg_{v_1}\Big(\Big((v_1 + \dots + v_k)^{q-1} + 1\Big)h'_1\Big) \geq q.
    \]
    One can readily verify that \( \deg_{v_i}(g_j) \leq q-1 \). Since \( \deg_{v_i}(g_j) \leq q-1 \), any monomial containing \( v_1^t \) for \( t \geq q \) must vanish through cancellation with terms in \( \sum_{j=2}^k \Big(v_j^q + v_j\Big)h'_j \). However, this implies that at least one \( h_j \) for \( 2 \leq j \leq k \) satisfies \( \deg_{v_1}(h_j) \geq q \), which is impossible since the algorithm of Theorem~\ref{theo:4Lev} ensures every \( v_1^t \) for \( t \geq q \) is reduced.

    Now consider the case \( v_1 = 0 \). Here, \( g_j = 0 \) and:
    \[
    \Big((v_2 + \dots + v_k)^{q-1} + 1\Big)h'_1 = 0.
    \]
    Thus, if \( v_2 + \dots + v_k = 0 \), then \( h'_1 = 0 \). By Theorem~\ref{theo:4Lev}, \( h'_1 \) can be expressed as:
    \[
    h'_1(v_2, \dots, v_k) = (v_2 + \dots + v_k)h''_2 + \sum_{j=3}^k (v_j^q - v_j)h''_j,
    \]
    where \( h''_2, \dots, h''_k \in \mathbb{F}_q[v_1, \dots, v_k] \) satisfy \( \deg(h'_i) \leq \deg(h'_1) - 1 \).

    Combining these results, we obtain:
    \[
    g_j(v_1, \dots, v_k) = \Big((v_1 + \dots + v_k)^{q-1} + 1\Big)\Big(v_2 + \dots + v_k\Big)h_1 + \sum_{j=2}^k \Big(v_j^q + v_j\Big)h_j,
    \]
    where \( \deg(h_i) \leq \deg(g_j) - q \) and \( h_1 \) is independent of \( v_1 \).
\end{proof}

Denote by \( h_j \) the polynomial \( h_1 \) of \( g_j \). The following corollary follows.

\begin{corollary}\label{cor:10}
    Conjecture~\ref{conj_2} holds if any of the following conditions are satisfied:
    \begin{enumerate}
        \item For all \( \beta_2, \dots, \beta_k \in \mathbb{F}_q \) such that \( \sum_{i=2}^k \beta_i \neq 0 \), there exists a polynomial \( h_j \) such that \( h_j(\beta_2, \beta_3, \dots, \beta_k) \neq 0 \).
        \item There exists a non-zero polynomial \( h_j(v_2, v_3, \dots, v_k) \) that is independent of one of the variables \( v_2, \dots, v_k \).
        \item There exists a non-zero polynomial \( h_j(v_2, v_3, \dots, v_k) \) such that for all \( 1 \leq i \leq k \), \( \deg_{v_i}(h_j) \leq 1 \).
    \end{enumerate}
\end{corollary}

\begin{proof}
    \begin{enumerate}
        \item Let \( \beta_1 = \sum_{i=2}^k \beta_i \neq 0 \), and denote \( \beta = (\beta_1, \dots, \beta_k) \). Then:
        \[
        g_j(\beta_1, \dots, \beta_k) = 1 \cdot \beta_1 \cdot h_j(\beta_2, \dots, \beta_k).
        \]
        If \( h_j(\beta_2, \dots, \beta_k) \neq 0 \), then \( f_{\beta} \not\equiv_R 0 \), and by Theorem~\ref{theo:5}, there is a solution for \( \beta \).

        \item Assume without loss of generality that \( h_j \neq 0 \) is independent of \( v_k \). 
        By Claim~\ref{claim:9_1} there exists \( \gamma = (\gamma_1, \dots, \gamma_k) \) such that \( g_j(\gamma_1, \dots, \gamma_k) \neq 0 \), leading to the fact that $h_j(\gamma_2,\dots,\gamma_k)\neq 0$.

        Substitute \( \gamma_1, \gamma_k \) with any non-zero \( \beta_1, \gamma'_k \) such that \( \beta_1 + \gamma'_k = \gamma_1 + \gamma_k \). Let \( \beta' = (\beta_1, \gamma_2, \gamma_3, \dots, \gamma'_k) \). By Lemma~\ref{lemma:12}, \( h_j \) is independent of \( v_1 \) and, by assumption, also independent of \( v_k \). Thus $g_j(\beta_1, \gamma_2, \gamma_3, \dots, \gamma'_k) \neq 0$ and therefore $f_{\beta'} \not\equiv_R 0$.
        By Theorem~\ref{theo:5}, there is a solution for \( \beta' \). 

        Since $f_v$ is permutation invariant in terms of $v$, we can find another $j'$ such that $g_{j'}(\gamma_2,\beta_1,\gamma_3,\dots,\gamma'_k)\neq 0$.
        We can then select any \( \beta_2, \gamma''_k \in \mathbb{F}_q \) such that \( \beta_2 + \gamma''_k = \gamma_2 + \gamma'_k \), and repeat the same steps as before.
        By continuing this process inductively, we prove that a solution exists for any \( \beta \) satisfying \( \sum_{i=1}^k \beta_i = 0 \).

        \item If \( h_j \) is independent of some variable, the result follows from the previous item. 
        Otherwise, let \( h_j = \sum_{i=2}^k c_i v_i \), where \( c_i \in \mathbb{F}_q \). 
        If all \( c_i \)'s are equal to $c\in \F_q$, then for any \( \beta = (\beta_1, \dots, \beta_k) \) such that $\beta_i$'s are non-zero, and  \( \sum_{i=1}^k \beta_i = 0 \), we have \( h_j(\beta_2, \dots, \beta_k) = c\sum_{i=2}^k\beta_i =  c\beta_1 \neq 0 \).
        By item 1, in this case, there is a solution for \( \beta \).

        Otherwise, assume without loss of generality that \( c_2 \neq c_3 \)
        We can also assume that \( \beta_2 \neq \beta_3 \), since the case where all \( \beta_i \)'s are equal has already been addressed in Corollary~\ref{cor:8}. 
        Then, if \( h_j(\beta_2, \beta_3, \dots, \beta_k) = 0 \), it follows that
        \begin{align*}
             h_j(\beta_3, \beta_2, \dots, \beta_k) &=  h_j(\beta_2, \beta_3, \dots, \beta_k) + c_2\beta_2+c_3\beta_3+ c_2\beta_3 + c_3\beta_2 \\
         &= (c_2+c_3)(\beta_2+\beta_3)\neq 0. 
        \end{align*}
        Since $f_v$ is permutation invariant in terms of $v$, and by item 1, there is a solution for \( \beta \).
    \end{enumerate}
\end{proof}

We will illustrate how to use items 2 and 3 of Corollary~\ref{cor:10} in Section~\ref{sec:2_3}. 

\section{The $s=2$ and $s=3$ Cases}\label{sec:2_3}

This section provides intuition for proving the general case of Theorem~\ref{theo:main2}.  
For $s=2$, we have $k = 2^{s-1} = 2$ and $\mathbb{F}_{2^2} = \text{GF}(4)$. Thus, 
\[
f_v(x_1, x_2) = v_1v_2(x_1 + x_2)(x_1 + x_2 + v_1)(x_1 + x_2 + v_2)(x_1 + x_2 + v_1 + v_2).
\]

Define
\begin{align*}
    g_1(v_1, v_2) &= v_1v_2(v_1v_2 + v_1^2 + v_2^2), \\
    g_2(v_1, v_2) &= v_1v_2(v_1^2v_2 + v_1v_2^2 + 1).
\end{align*}
We begin with the following claim:

\begin{claim}\label{claim:11}
    For $k=2$, the following holds:
    \[
    f_v(x_1, x_2) \equiv_R g_1(v_1, v_2)(x_1 + x_2)^2 + g_2(v_1, v_2)(x_1 + x_2).
    \]
\end{claim}

\begin{proof}
    By definition:
    \begin{align*}
        f_v(x_1,x_2) &= v_1v_2(x_1+x_2)(x_1+x_2+v_1)(x_1+x_2+v_2)(x_1+x_2+v_1+v_2) \\
        &= v_1v_2\cdot\Big((x_1+x_2)^4 + c_3(x_1+x_2)^3 + c_2(x_1+x_2)^2 + c_1(x_1+x_2) + c_0\Big),
    \end{align*}
    where the coefficients $c_3, c_2, c_1$, and $c_0$ are defined as:
    \begin{align*}
        c_3 &= v_1 + v_2 + (v_1 + v_2) = 0, \\
        c_2 &= v_1v_2 + v_1(v_1 + v_2) + v_2(v_1 + v_2) = v_1v_2 + v_1^2 + v_2^2, \\
        c_1 &= v_1v_2(v_1 + v_2), \\
        c_0 &= 0.
    \end{align*}

    In $GF(4)$, we have the property \( x_i^4 \equiv_R x_i \). Substituting this into the polynomial, we get:
    \begin{align*}
        f_v(x_1, x_2) &\equiv_R v_1v_2(v_1v_2 + v_1^2 + v_2^2)(x_1 + x_2)^2 + v_1v_2(v_1^2v_2 + v_1v_2^2 + 1)(x_1 + x_2).
    \end{align*}

    Setting:
    \[
    g_1(v_1, v_2) = v_1v_2(v_1v_2 + v_1^2 + v_2^2), \quad g_2(v_1, v_2) = v_1v_2(v_1^2v_2 + v_1v_2^2 + 1),
    \]
    we conclude:
    \[
    f_v(x_1, x_2) \equiv_R g_1(v_1, v_2)(x_1 + x_2)^2 + g_2(v_1, v_2)(x_1 + x_2).
    \]
\end{proof}

\begin{corollary}
    \( f_v \equiv_R 0 \) if and only if \( v_1 = v_2 \) for nonzero \( v_i \)'s.
\end{corollary}

\begin{proof}
    From Example~\ref{examp:1}, we know
    \[
    g_2(v_1, v_2) = \Big((v_1 + v_2)^3 + 1\Big)v_2^2 + \Big(v_2^4 + v_2\Big)(v_1 + v_2).
    \]
    The term \(\Big(v_2^4 + v_2\Big)(v_1 + v_2)\) vanishes for every \(v_2\). The term \(\Big((v_1 + v_2)^3 + 1\Big)v_2^2\) vanishes if and only if \(v_1 \neq v_2\) in \(\text{GF}(4)\) for nonzero \(v_i\)'s.
\end{proof}

For $s=3$, we have:
\[
f_v(x_1, x_2, x_3, x_4) = v_1 \prod_{1 \leq i < j \leq 4} 
    (x_i + x_j)(x_i + x_j + v_i)(x_i + x_j + v_j)(x_i + x_j + v_i + v_j),
\]
where we omit multiplying by $v_2v_3v_4$ to simplify the calculations.

It can be shown combinatorially that the monomial $v_1^7v_4^2x_1^6x_2^6x_3^4$ is generated 15 times, and thus it does not vanish.  
By Lemma~\ref{lemma:12}, the coefficient of $x_1^6x_2^6x_3^4$ can be expressed as:
\[
\Big((v_1 + v_2 + v_3 + v_4)^7 + 1\Big)(v_2 + v_3 + v_4)h_1(v_2, v_3, v_4) + \sum_{j=2}^4 \Big(v_j^q + v_j\Big)h_j.
\]

The degree of $h_1$ is $25 - (6 + 6 + 4 + 8) = 1$. By item 3 of Corollary~\ref{cor:10}, there exists a solution for any $\alpha_1, \alpha_2, \alpha_3, \alpha_4$ such that \(\sum_{i=1}^4 \alpha_i = 0\).  
Using the division algorithm, it can be verified that \(h_1(v_2, v_3, v_4) = v_3 + v_4\), which is independent of \(v_2\).  
Thus, item 2 of Corollary~\ref{cor:10} can also be applied.

\section{Conclusion}
\label{sec:conc}

In this work, we have explored the properties and requirements of functional $k$-batch codes of dimension $s$, which are characterized by $n$ servers storing linear combinations of $s$ linearly independent information bits. We addressed a recent conjecture suggesting that for any set of $k = 2^{s-1}$ requests, the optimal solution requires $2^s - 1$ servers.
A new approach to this problem was proposed using the special polynomial $f_v$. The proposed technique is based on the structure of the polynomial $f_v$, which was analyzed in detail in this paper. We demonstrated several examples for the cases $s=2$ and $s=3$. The calculations become more complex for $s \geq  4$ and are left for future work.
To further the understanding of this conjecture, we introduced several sufficient conditions that provide new insights into the construction and optimization of functional $2^{s-1}$-batch codes.






\end{document}